\documentclass[11pt, letterpaper]{article}
\usepackage[utf8]{inputenc}
\usepackage[left=1in, top=1in, right=1in, bottom=1in]{geometry}
\usepackage[colorlinks,citecolor=blue]{hyperref}
\usepackage{amsmath, amsthm, amsfonts, amssymb}
\usepackage{mathtools}
\usepackage{bm, bbm}
\usepackage{xcolor}
\usepackage{mdframed}
\usepackage{graphicx}
\usepackage{comment}
\usepackage[ruled, linesnumbered]{algorithm2e}
\usepackage{cleveref}
\usepackage{enumitem}
\usepackage{natbib}
\usepackage{color-edits}
\usepackage{setspace}
\newtheorem{theorem}{Theorem}[section]
\newtheorem{proposition}[theorem]{Proposition}
\newtheorem{lemma}[theorem]{Lemma}
\newtheorem{definition}[theorem]{Definition}
\newtheorem{claim}[theorem]{Claim}

\newtheorem{corollary}[theorem]{Corollary}

\DeclareMathOperator*{\argmax}{argmax}

\title{Learning to Coordinate Bidders in Non-Truthful Auctions}

\author{
Hu Fu\thanks{Shanghai University of Finance and Economics, School of Computing and Artificial Intelligence, Key Laboratory of Interdisciplinary Research of Computation and Economics.  \texttt{fuhu@mail.shufe.edu.cn}.}
\and
Tao Lin\thanks{Microsoft Research \& The Chinese University of Hong Kong, Shenzhen. \texttt{lintao@cuhk.edu.cn}. }
} 

\newcommand{\AutoAdjust}[3]{\mathchoice{ \left #1 #2  \right #3}{#1 #2 #3}{#1 #2 #3}{#1 #2 #3} }
\newcommand{\Xcomment}[1]{{}}

\newcommand{\InBrackets}[1]{\AutoAdjust{[}{#1}{]}}
\newcommand{\Ex}[2][]{\operatorname{\mathbb E}_{#1}\InBrackets{#2}}
\newcommand{\Prx}[2][]{\operatorname{\Pr}_{#1}\InBrackets{#2}}

\newcommand{\given}{\;\mid\;}

\newcommand{\eps}{\varepsilon}
\newcommand{\reals}{\mathbb R}
\newcommand{\E}{\mathbb E}

\newcommand{\noaccents}[1]{#1}
\newcommand{\newagentvar}[3][\noaccents]{%
\expandafter\newcommand\expandafter{\csname #2\endcsname}{#1{#3}}%
\expandafter\newcommand\expandafter{\csname #2s\endcsname}{#1{\boldsymbol{#3}}}%
\expandafter\newcommand\expandafter{\csname #2smi\endcsname}[1][i]{#1{\boldsymbol{#3}}_{-##1}}%
\expandafter\newcommand\expandafter{\csname #2i\endcsname}[1][i]{#1{#3}_{##1}}%
\expandafter\newcommand\expandafter{\csname #2ith\endcsname}[1][i]{#1{#3}_{(##1)}}%
}

\newcommand{\newvecagentvar}[3][\noaccents]{%
\expandafter\newcommand\expandafter{\csname #2\endcsname}{#1{\boldsymbol{#3}}}%
\expandafter\newcommand\expandafter{\csname #2s\endcsname}{#1{\boldsymbol{#3}}}%
\expandafter\newcommand\expandafter{\csname #2smi\endcsname}[1][i]{#1{\boldsymbol{#3}}_{-##1}}%
\expandafter\newcommand\expandafter{\csname #2i\endcsname}[1][i]{#1{\boldsymbol{#3}}_{##1}}%
\expandafter\newcommand\expandafter{\csname #2ith\endcsname}[1][i]{#1{#3}_{(##1)}}%
}

\newagentvar{type}{t}
\newagentvar{alloc}{x}
\newagentvar{pay}{p}
\newagentvar{typespace}{T}
\newagentvar{val}{v}
\newagentvar{util}{u}
\newagentvar{payment}{p}
\newagentvar{dist}{D}
\newagentvar{bid}{b}
\newagentvar{strat}{\sigma}
\newagentvar{stratspace}{\Sigma}

\newagentvar{expostU}{U}
\newagentvar{empDist}{E}

\newcommand{\interimalloc}{x}
\newcommand{\interimpay}{p}

\newcommand{\nsample}{m}
\newcommand{\sampleSet}{\mathcal S}
\newagentvar{sample}{x}
\newcommand{\alg}{\mathcal{A}}
\newcommand{\strategyset}{\mathcal{B}}

\newcommand{\posdist}{\dist^+}
\newcommand{\negdist}{\dist^-}

\newcommand{\distconst}{c_1}
\newcommand{\posbid}{\strat^+}
\newcommand{\negbid}{\strat^-}
\DeclareMathOperator{\kl}{KL}

\newcommand{\diffalg}{\mathcal{H}}
\newcommand{\diffconst}{c_2}
\newcommand{\diffset}{\theta}
\newcommand{\coordsets}{\mathbf{S}}

\DeclareMathOperator{\emp}{Emp}
\DeclareMathOperator{\empp}{Empp}
\newcommand{\func}{h}
\newcommand{\funcClass}{\mathcal{H}}
\newcommand{\PinputSpace}{\mathcal{X}}
\newagentvar{Pinput}{x}
\newagentvar{Plabel}{l}
\newagentvar{Pwitness}{t}
\DeclareMathOperator{\sgn}{sgn}
\DeclareMathOperator{\Pdim}{Pdim}
\newagentvar{permSample}{{\tilde{x}}}

\begin{document}

\maketitle

\begin{abstract}
In non-truthful auctions such as first-price and all-pay auctions, the independent strategic behaviors of bidders, with the corresponding Bayes-Nash equilibrium notion, are notoriously difficult to characterize and can cause undesirable outcomes. 
An alternative approach to achieve better outcomes in non-truthful auctions is to coordinate the bidders: let a mediator make incentive-compatible recommendations of correlated bidding strategies to the bidders, namely, implementing a Bayes correlated equilibrium (BCE).
The implementation of BCE, however, requires knowledge of the distributions of bidders' private valuations, which is often unavailable. 
We initiate the study of the sample complexity of learning Bayes correlated equilibria in non-truthful auctions.
We prove that the set of strategic-form BCEs in a large class of non-truthful auctions, including first-price and all-pay auctions, can be learned with a polynomial number $\tilde O(\frac{n}{\varepsilon^2})$ of samples of bidders' values.
This moderate number of samples demonstrates the statistical feasibility of learning to coordinate bidders. 
Our technique is a reduction to the problem of estimating bidders' expected utility from samples, combined with an analysis of the pseudo-dimension of the class of all monotone bidding strategies.
\end{abstract}

\section{Introduction}
Non-truthful auctions, among which the most ubiquitous and fundamental example is perhaps the first-price auction (FPA), are gaining more popularity over truthful auctions (such as the second-price auction) in the online advertising markets in recent years, due to reasons like transparency and market competition \citep{AL18, fpa1, fpa2, paes_leme_why_2020, goke_bidders_2022}.

Non-truthful auctions require strategic bidding, and the classical equilibrium notion for such games is \emph{Bayes Nash equilibrium (BNE)}, where bidders with private valuations choose bidding strategies independently to maximize their own payoffs. 
In auctions like FPA, however, the independent choices of strategies by bidders, with the corresponding BNE notion, are known to be problematic. 
For example, the independent learning dynamics of bidders do not always converge to BNE; they often oscillate, cause undesirable instability to the system, or lead to outcomes with low welfare or low revenue \citep{edelman_strategic_2007, deng_nash_2022, banchio_artificial_2022, paes_leme_complex_2024, bichler_beyond_2025}. 
When bidders' private valuations are not independent and identically distributed, the BNE of FPA is notoriously difficult to characterize or compute \citep{chen_complexity_2023, filos-ratsikas_computation_2024}. 
In such scenarios, BNE is no longer a good prediction for bidders' behavior or a desired outcome for the designer of an auction system. 


An alternative, potentially more desirable, equilibrium concept for non-truthful auctions is \emph{Bayesian correlated equilibrium (BCE)}.
Bidders in real-world auctions sometimes communicate with each other before deciding on their bids, in which case BNE is not an appropriate equilibrium concept because it assumes independent decisions.
The equilibrium outcome of a game where players communicate before making decisions is equivalent to a correlated equilibrium \citep{forges_approach_1986, myerson_game_1991}.
Moreover, players' communication can be facilitated by a central coordinator. 
In modern online advertising auctions, bidders delegate their bidding tasks to third-party platforms that run auto-bidding algorithms \citep{aggarwal_auto-bidding_2024}, and those platforms can, at least in principle, coordinate different bidders' bids \citep{decarolis_bid_2023}. 
Such coordination might increase bidders' welfare (because the set of BCEs is weakly larger than the set of BNEs), as well as stabilize the auction system by avoiding the possibly chaotic dynamics caused by independent auto-bidding \citep{paes_leme_complex_2024}.

Although BCE might be more desirable than BNE for bidders, third-party bidding platforms, and auction system designers, the implementation of a BCE is not easy.
To implement a BCE, the coordinator has to recommend correlated bidding strategies to bidders in an \emph{incentive-compatible} way: conditioning on the recommendation, each bidder should be willing to use the recommended strategy. 
To achieve this goal, the coordinator must know the expected payoff to each bidder under the recommended strategy. 
While players' payoffs are easily computable in complete-information games, it is not the case in incomplete-information games like auctions, where players (bidders) have private types (valuations) that affect their actions (bids), so each player's expected payoff depends on the distributions of other players' private types. 
To implement a BCE in auctions, the coordinator has to know the bidders' value distributions, but such distributions are hardly available in practice, because samples of bidders' private values are often limited and costly to obtain. 

Motivated by the intricate dependency of BCE on players' type distributions, this work aims to characterize the \emph{sample complexity} of BCE in non-truthful auctions. Our main question is: \emph{How many samples from bidders' value distributions do we need to reconstruct the set of all BCEs in a non-truthful auction (such as first-price auction)?} 
Note that our goal is to learn the set of all BCEs, not just one BCE.
Learning the set of all BCEs allows the coordinator to pick the best BCE to implement, according to any definition of ``best''.

\paragraph{Overview of Our Contributions.}
We initiate the study of the sample complexity of Bayesian correlated equilibrium in non-truthful auctions.
In our model, $n$ bidders have independent but not identically distributed private values $v_1 \sim \disti[1], \ldots, v_n \sim \disti[n]$. 
As there are multiple notions of BCEs in the literature \citep{forges_correlated_2006}, we focus on the strategic-form BCE, where a coordinator samples a profile of bidding strategies $\strats = (\strati[1], \ldots, \strati[n])$ from a joint distribution $Q$ and privately recommends each strategy $\strati$ to the corresponding bidder $i$, such that every bidder is willing to obey the recommended strategy given their posterior belief about others' strategies. 
We show that the set of strategic-form $\eps$-BCEs of a large class of non-truthful auctions (including first-price and all-pay auctions) can be learned with a polynomial number of samples of bidders' values: $\tilde O(\frac{n}{\eps^2})$, where $n$ is the number of bidders. (The $\tilde O(\cdot)$ notation omits logarithmic factors.)
Formally presented in Theorem \ref{theorem:find-BCE}, our sample complexity result holds for any distributions of bidders' private values that are bounded and independent. 
The moderate $\tilde O(\frac{n}{\eps^2})$ sample complexity demonstrates the statistical feasibility of learning to coordinate bidders in non-truthful auctions in data-scarce scenarios.

\paragraph{Overview of Our Techniques.}
A straightforward attempt to bound the sample complexity of BCE is to analyze the number of samples needed to learn the bidders' value distributions $D_1, \ldots, D_n$.  If we could learn the distributions with accuracy $\eps$, then the BCE computed from the learned distributions would be an $\eps$-BCE on the true distributions.  However, since each $D_i$ is an arbitrary distribution on a continuous support (the space of value), it is impossible to learn such a distribution accurately using a small number of samples.

Instead, our approach to proving the $\tilde O(\frac{n}{\eps^2})$ sample complexity of BCE is a reduction to the \emph{utility estimation} problem: given a set of joint bidding strategies $\strategyset$, we aim to estimate each bidder's expected payoff under any joint bidding strategy $\strats$ in $\strategyset$, using samples of bidders' values. 
We show that, when $\strategyset$ consists of all bidding strategies that are monotonically increasing, each bidder's expected payoff can be estimated using $\tilde O(\frac{n}{\eps^2})$ samples. 
Interestingly, the utility estimation problem cannot be solved with finitely many samples if $\strategyset$ contains all (possibly non-monotone) bidding strategies (shown in Proposition \ref{prop:utility-estimation-unsolvable}).
But fortunately, strategic-form BCEs are always monotone (an interesting observation we prove in Proposition \ref{prop:monotone}), so it suffices to estimate utilities for monotone bidding strategies only.
As a result, the set of $\eps$-BCEs can be learned using the same amount of samples, $\tilde O(\frac{n}{\eps^2})$. 
To prove the aforementioned sample complexity of the utility estimation problem, we use a technique from statistical learning theory, involving a non-trivial analysis of the pseudo-dimension of bidders' utility functions under monotone strategies.
We also prove an almost matching lower bound: $\Omega(\frac{n}{\eps^2})$ samples are necessary to estimate bidders' expected utilities for all monotone strategies. 

In addition to BCEs, our utility estimation result also implies a sample complexity for learning BNEs (Theorem \ref{theorem:find-BNE}). Similarly to BCEs, BNEs in non-truthful auctions like FPA are also monotone (this is an observation by \cite{maskin_equilibrium_2000} and a corollary of our Proposition \ref{prop:monotone}), so the expected utilities of the bidders under BNEs can be estimated with $\tilde O(\frac{n}{\eps^2})$ samples as well. Therefore, the sample complexity of learning the set of all BNEs is also upper bounded by $\tilde O(\frac{n}{\eps^2})$.  

\subsection{Related Works}
\paragraph{Coordination in Auctions.}
Bidder's coordination, or collusion, is a long-standing topic in the traditional auction literature \citep{graham_collusive_1987, mailath_collusion_1991, mcafee_bidding_1992, marshall_bidder_2007, hendricks_bidding_2008, lopomo_bidder_2011} and has recently been studied in the online ad auction domain as well \citep{decarolis_marketing_2020, romano_power_2022, decarolis_bid_2023, chen_coordinated_2023}. Contrary to the previous view that collusion can undermine the auctioneer's revenue, we take a positive viewpoint here: coordination might be desirable for the system designer. This is because: (1) coordination might prevent the unstable strategizing behaviors of independent bidders; (2) the set of BCEs is larger than the set of BNEs in theory, so the coordinator can potentially induce an equilibrium with a (weakly) higher revenue or welfare than any independent equilibrium. 

\paragraph{Bayesian Correlated Equilibrium in General Games.}
Incentivizing bidders to coordinate is equivalent to finding a Bayes correlated equilibrium in the auction game.
The classical notion of correlated equilibrium \citep{aumann_subjectivity_1974} is defined for complete information games.
For incomplete information games like auctions with private values, the literature has defined multiple notions of Bayes correlated equilibria, such as strategic-form BCE, agent-normal-form BCE, and communication equilibrium \citep{myerson_optimal_1982, forges_correlated_2006, bergemann_bayes_2016, fujii_bayes_2023}.
We consider the strategic-form BCE \citep{forges_correlated_2006} where the coordinator recommends randomized joint bidding strategies to all bidders without knowing the bidders' private values.
This type of BCE satisfies monotonicity (as we will show in Proposition \ref{prop:monotone}) and does not alter any bidder's belief about other bidders' private values. These two crucial properties ensure the learnability of the BCE when bidders' value distribution is unknown. 

\paragraph{Sampling from Value Distributions.}
An assumption of our work is that the learner has sample access to the underlying distribution of bidders' values. 
This is a standard assumption in the literature of learning in mechanism design \citep[e.g.][]{CR14, MR15, MR16, BSV16, BSV18, GN17, Syrgkanis17, GW18, guo2019settling, brustle_multi-item_2020, yang_learning_2021, guo_generalizing_2021}.
While most of those works study revenue maximization in truthful auctions, we consider the under-explored problems of utility estimation and equilibrium learning in non-truthful auctions. 

Value samples have been assumed in the context of learning in non-truthful auctions \citep{BSV19, vitercik2021automated}.  
Just as in classical microeconomics, prior knowledge (in the form of samples here) comes from market research, survey, simulation etc., and is not assumed to be from past bidding history.  
We distance our approach from the line of work on learning non-truthful auctions where samples are from past bidding history \citep{CHN17, HT19}.  
This latter approach, with obvious merits, has its limitations.  
Crucially, it assumes that the observed bidding in a non-truthful auction is at equilibrium, which may not be the case in reality.
Also, to avoid strategic issues between auctions, the bidders need to be short-lived or myopic.  
The two approaches (value samples vs.\@ bid samples) complement each other even in learning problems for non-truthful auctions. This work takes the first approach, and leaves the direction with bid samples as an enticing open question.

\paragraph{Utility Estimation in Games.}
Given a non-truthful auction, \cite{BSV19} studied the number of value samples needed to learn the maximal utility a bidder could gain by non-truthful bidding, when all other bidders are truthful.    
In comparison, we learn utilities when all bidders use arbitrary monotone bidding strategies; this suffices for the study of virtually all properties of an auction, including the task of \cite{BSV19}.\footnote{
Our results imply that the maximal utility (w.r.t~the opponents' value distribution) obtained by non-truthful bidding can be approximated by the maximal obtainable utility w.r.t.~the empirical distribution, which can be computed by enumerating the samples in the empirical distribution because a best-responding bid must be equal to (or slightly more than) some opponent's value from the empirical distribution. 
}


\cite{AGCU19, marchesi_learning_2020, duan_is_2023} studied utility estimation and equilibrium learning for general normal-form games with random utility matrices (with sample access).  Similar to us, they frame the problem as a PAC learning problem and bound the number of samples using complexity measures (e.g., Rademacher complexity, covering number) of some function classes.  But they did not characterize those complexity measures for specific games.  Bounding those complexity measures is generally challenging, if not impossible.  
For example, without monotonicity in bidding strategies, the pseudo-dimension of the utility functions in first-price auctions is unbounded, as implied by our Proposition \ref{prop:utility-estimation-unsolvable}.

\paragraph{Equilibrium Computation in Auctions.}
There is a large literature on the computation of equilibrium in non-truthful auctions \citep[e.g.][]{Marshall94, fibich_asymmetric_2003, Gayle08, escamocher_existence_2009, wang_bayesian_2020, chen_complexity_2023, filos-ratsikas_complexity_2021, filos-ratsikas_computation_2024}.  
Although there has been major progress on the computation of BNE in first-price auctions with common prior distributions \citep{wang_bayesian_2020, chen_complexity_2023}, this problem turns out to be PPAD-hard with subjective prior distributions \citep{filos-ratsikas_complexity_2021}. 
On the other hand, a BCE is known to be easier to compute than a BNE \citep{filos-ratsikas_computation_2024}, which provides an additional motivation for us to study the BCE of non-truthful auctions. 

\section{Preliminary: Auctions, BNE, and BCE}
\label{FPA_sample:prelim}

\paragraph{Auctions.}
Consider a single-item auction with $n$ bidders denoted by $[n] = \{1, \ldots, n\}$.
Each bidder~$i\in[n]$ has a private value $\vali$ drawn from a distribution $\disti$ supported on $\typespacei \subseteq [0, H] \subseteq \reals_+$, where $H$ is an upper bound on the bidder's value.
The size of the support $|\typespacei|$ can be infinite.
Different bidders' values are independent and can be non-identically distributed, so the joint value distribution $\dists = \prod_{i=1}^n \disti$ is a product distribution. 
Each bidder~$i$ makes a sealed-envelope bid of~$\bidi \in [0, H]$. The auction maps the vector of bids~$\bids = (\bidi[1], \ldots, \bidi[n])$ to \emph{allocation} and \emph{payments}, where allocation  
$\alloci(\bids) \in [0, 1]$ is the probability with which bidder~$i$ receives the item, with $\sum_{i=1}^n \alloci(\bids) \leq 1$, and payment $\payi(\bids)$ is the payment made by bidder~$i$ to the auctioneer.
%
Bidder~$i$'s \emph{ex post utility} is denoted by
\begin{equation}\label{eq:ex-post-utility}
\expostUi(\vali, \bids) := \vali \alloci(\bids) - \payi(\bids).
\end{equation}

We focus on the allocation rule where the bidder with the highest bid wins, with ties broken randomly: $\alloci(\bids) = \frac{\mathbb{I}[\bidi = \max_{j\in[n]} \bidi[j]]}{|\argmax_{j\in[n]} \bidi[j] |}$. Auctions with reserve prices can be modeled by adding an additional bidder who always bids the reserve price. 


We consider any payment function of the following form:
\begin{align}
\payi(\bids) = \alloci(\bids) f_i(\bidi) + g_i(\bidi). 
\end{align}
where functions $f_i$ and $g_i$ satisfy $0 \le f_i(\bidi), g_i(\bidi) \le H$.
For example,
\begin{itemize}
    \item in the \emph{first-price auction} (FPA), the highest bidder wins the item and pays her bid, and other bidders pays zero: $\payi^{\mathrm{FPA}}(\bids) = \alloci(\bids) \bidi$.  
    \item In the \emph{all-pay auction} (APA), the highest bidder wins the item but all bidders pay their bids: $\payi^{\mathrm{APA}}(\bids) = \bidi$. APA is a good model for, e.g., crowdsourcing \citep{CHS12}. 
\end{itemize} 
We assume that, fixing the bids $\bidsmi$ of other bidders, if bidder $i$ wins the item at bid $\bidi$, then her payment must be strictly increasing in her bid: $\forall \bidi' > \bidi$, 
\begin{equation} \label{condition:payment-strictly-increasing}
    \alloci(\bidi, \bidsmi) > 0 ~ \implies~  \payi(\bidi', \bidsmi) > \payi(\bidi, \bidsmi). 
\end{equation}
This condition is satisfied by both first-price and all-pay auctions. 

\paragraph{Strategies and Equilibria.}
A \emph{(bidding) strategy $\strati : \typespacei \to [0, H]$} is a mapping from the bidder's value $\vali$ to bid $\bidi = \strati(\vali)$. Let $\stratspacei = T_i\to [0, H]$ be the strategy space of bidder $i$, and let $\stratspaces = \prod_{i=1}^n \stratspacei$ be the joint strategy space of all bidders.   
Let $\bids = \strats(\vals) = (\strati[1](\vali[1]), \ldots, \strati[n](\vali[n]))$ denote the bids of all bidders, and $\bidsmi = \stratsmi(\valsmi)$ denote the bids of bidders except $i$. 
When other bidders use strategies $\stratsmi$, bidder $i$ with value~$\vali$ and bid $\bidi$ obtains \emph{interim utility} 
\begin{align}
\utili(\vali, \bidi, \stratsmi) := & \Ex[\valsmi\sim \distsmi]{\expostUi(\vali, \bidi, \stratsmi(\valsmi))}     
 \label{eq:interim-util} \\
 = & \Ex[\valsmi \sim \distsmi]{\vali \alloci(\bidi, \stratsmi(\valsmi)) - \payi(\bidi, \stratsmi(\valsmi))}. \nonumber
\end{align}

We define Bayes Nash equilibrium for the auction game:  
\begin{definition}[Bayes Nash equilibrium]
\label{def:bne}
For $\eps \geq 0$, a joint bidding strategy $\strats = (\strati[1], \ldots, \strati[n])$ is a (pure-strategy) \emph{$\eps$-Bayes Nash equilibrium ($\eps$-BNE)} for value distribution $\dists = \prod_{i=1}^n \disti$ if for each bidder~$i\in[n]$, any value $\vali \in \typespacei$, any bid $\bidi' \in [0, H]$,
\begin{align*}
\utili(\vali, \strati(\vali), \stratsmi) \, \geq \, 
\utili(\vali, \bidi', \stratsmi) - \eps.
\end{align*}
When $\eps = 0$, $\strats$ is a \emph{Bayes Nash equilibrium (BNE)}.
\end{definition}

We also define a Bayes correlated equilibrium for the auction game. Correlated equilibria are typically defined for complete information games.  For incomplete information games like auctions, there are multiple definitions of Bayes correlated equilibria in the literature \citep{forges_correlated_2006}. We consider the ``strategic-form Bayes correlated equilibrium'' in \cite{forges_correlated_2006}, which regards the incomplete information game as a normal form game where a player's pure strategy is the mapping $\strati$. A \emph{correlation device}, or \emph{mediator}, can sample a joint strategy $\strats = (\strati[1], \ldots, \strati[n])$ from a joint distribution $Q \in \Delta(\stratspaces)$, and recommend each strategy $\strati$ to the respective bidder $i$, while ensuring that no bidder has incentive to deviate from the recommended strategy.  
\begin{definition}[Bayes correlated equilibrium]
\label{def:bce}
For $\eps \geq 0$, a distribution $Q \in \Delta(\stratspaces)$ over joint bidding strategies is an \emph{$\eps$-Bayes correlated equilibrium ($\eps$-BCE)} for value distribution $\dists = \prod_{i=1}^n \disti$ if for each bidder~$i\in[n]$, any value $\vali \in \typespacei$, any deviation function $\phi_i : \stratspacei \times \typespacei \to [0, H]$,
\[\E_{\strats \sim Q} \big[ \utili(\vali, \strati(\vali), \stratsmi) \big] \, \geq \, 
\E_{\strats \sim Q} \big[ \utili(\vali, \phi_i(\strati, \vali), \stratsmi) \big] - \eps.\]
When $\eps = 0$, $Q$ is a \emph{Bayes correlated equilibrium (BCE)}.
\end{definition}

A pure-strategy BNE is a BCE where bidders' joint strategy $\strats = (\strati[1], \ldots, \strati[n])$ is deterministic. 
A mixed-strategy BNE is a BCE where bidders' strategies $\strati[1], \ldots, \strati[n]$ are randomized and independent. 

We say a bidding strategy $\strati$ is \emph{monotone} if it is weakly increasing: $\val \geq \val' \Rightarrow \strati(\val) \geq \strati(\val')$, 
A joint bidding strategy $\strats$ is monotone if all individual strategies $\strati[1], \ldots, \strati[n]$ are monotone. 
A BCE $Q \in \Delta(\stratspaces)$ is monotone if every joint strategy $\strats$ sampled from $Q$ is monotone. 
\citep{maskin_equilibrium_2000} show that the BNEs of auction games are ``essentially monotone''. We generalize their result to BCEs. 

\begin{proposition}
\label{prop:monotone}
Under Assumption \eqref{condition:payment-strictly-increasing}, any BCE $Q \in \Delta(\stratspaces)$ of the auction game is ``essentially monotone'' in the following sense: for almost every joint strategy $\strats = (\strati[1], \ldots, \strati[n])$ sampled from $Q$, every bidder $i$'s strategy $\strati(\vali)$ is weakly increasing except when $\vali$ is too low that bidder $i$ wins the item with probability $0$. 
\end{proposition}

\begin{proof}
Let $Q \in \Delta(\stratspaces)$ be a BCE, with $\strats \sim Q$. Suppose bidder $i$'s strategy $\strati$ is not weakly increasing on two values $\vali < \vali'$, namely, $\bidi = \strati(\vali) > \bidi' = \strati(\vali')$. By the definition of BCE, conditioning on bidder $i$ being recommended $\strati$, we have
\begin{align*}
\Ex[\stratsmi | \strati]{\utili(\vali, \bidi, \stratsmi)} \, \geq \, 
\Ex[\stratsmi | \strati]{\utili(\vali, \bidi', \stratsmi)}. 
\end{align*}
Define interim allocation $\alloci(\bidi) = \E_{\stratsmi|\strati}\Ex[\valsmi \sim \distsmi]{\alloci(\bidi, \stratsmi(\valsmi))}$
and interim payment $\paymenti(\bidi) = \E_{\stratsmi|\strati} \Ex[\valsmi \sim \distsmi]{\paymenti(\bidi, \stratsmi(\valsmi))}$. Then we have
\begin{equation}\label{eq:monotone-maximize-1}
    \vali \interimalloc_i(\bidi) - \interimpay_i(\bidi) ~ \ge ~ \vali \interimalloc_i(\bidi') - \interimpay_i(\bidi'). 
\end{equation} 
Switching the roles of $\vali$ and $\vali'$, 
\begin{equation}\label{eq:monotone-maximize-2}
    \vali' \interimalloc_i(\bidi') - \interimpay_i(\bidi') ~ \ge ~ \vali' \interimalloc_i(\bidi) - \interimpay_i(\bidi). 
\end{equation} 
Adding \eqref{eq:monotone-maximize-1} and \eqref{eq:monotone-maximize-2}, we obtain
\begin{equation*}
    \big(\vali' - \vali\big)\cdot \big[ \interimalloc_i(\bidi') - \interimalloc_i(\bidi) \big] ~ \ge ~ 0. 
\end{equation*}
Since $\vali' > \vali$, we obtain $\interimalloc_i(\bidi') \ge \interimalloc_i(\bidi)$.
But under the assumption of $\bidi > \bidi'$, we have $\interimalloc_i(\bidi') \le \interimalloc_i(\bidi)$ because the function $\interimalloc_i(\cdot)$ is weakly increasing. 
Therefore, it must be
\begin{equation}\label{eq:monotone-equal-alloc}
    \interimalloc_i(\bidi') = \interimalloc_i(\bidi). 
\end{equation}
Plugging \eqref{eq:monotone-equal-alloc} into \eqref{eq:monotone-maximize-1} and \eqref{eq:monotone-maximize-2}, we obtain
\begin{equation*} 
    \interimpay_i(\bidi') = \interimpay_i(\bidi). 
\end{equation*}
If $\interimalloc_i(\bidi') = \interimalloc_i(\bidi) > 0$, then we have $\interimpay_i(\bidi) > \interimpay_i(\bidi')$ by Assumption \eqref{condition:payment-strictly-increasing}, which leads to a contradiction.  So, it must be $\interimalloc_i(\bidi') = \interimalloc_i(\bidi) = 0$, which means that bidder $i$ never wins the item under values $\vali$ and $\vali'$. 
%
%
\end{proof}

Since BCE is essentially monotone and any essentially monotone BCE can be converted to a monotone BCE without affect any bidder's expected utility, we will restrict attentions to monotone BCE. The set of monotone $\eps$-BCEs depends on the bidders' value distribution $\dists$.  We denote this set by \[\mathrm{BCE}(\dists, \eps) = \big\{ Q \in \Delta(\stratspaces) ~ \big| ~ Q \text{ is monotone and is an $\eps$-BCE on value distribution $\dists$} \big\}.\]

\paragraph{Our goal:}
Our goal is to learn the set $\mathrm{BCE}(\dists, \eps)$ when bidders' value distribution $\dists$ is unknown and can only be accessed by sampling. We aim to characterize the number of samples that are needed to achieve this goal. 

\section{Sample Complexity of Estimating Utility}
\label{FPA_sample:fpa}
A crucial step to learn the set of BCEs in an auction with unknown distribution $\dists$ is to estimate the bidders' expected utility for any given joint bidding strategy $\strats$.
We call this problem \emph{utility estimation}. 
The utility estimation problem is also interesting by itself, so we study the sample complexity of utility estimation in this section. 

Formally, we are given a set of $\nsample$ samples $\mathcal S = \{\vals^{(1)}, \ldots, \vals^{(\nsample)}\}$ from the value distribution $\dists = \prod_{i=1}^n \disti$, where each sample $\vals^{(j)} = (\vali[1]^{(j)}, \ldots, \vali[n]^{(j)})$ contains the values of all bidders, we aim the estimate the expected utility of every bidder under every possible joint bidding strategy $\strats$.
A \emph{utility estimation algorithm}, denoted by $\alg$, takes the samples $\mathcal S$, bidder index $i$, value $\vali$, and all bidders' strategies $\strats$ as input,  
outputs $\alg(\mathcal S, i, \vali, \strats)$ to estimate bidder~$i$'s interim utility $\utili(\vali, \strati(\vali), \stratsmi) = \Ex[\valsmi\sim \distsmi]{\expostUi(\vali, \strati(\vali), \stratsmi(\valsmi))}$. 

\begin{definition}[utility estimation]
\label{def:util-learn-ensemble}
Let $\strategyset \subseteq \stratspaces$ be a set of joint bidding strategies. 
For $\eps>0, \delta \in (0, 1)$, we say an algorithm~$\alg$ \emph{$(\eps, \delta)$-estimates with $\nsample$~sample the utilities over $\strategyset$} if,
for any value distribution~$\dists$,
with probability at least $1 - \delta$ over the random draw of $m$ samples from $\dists$, 
for any joint bidding strategy $\strats \in \strategyset$, 
for each bidder~$i\in[n]$ and any value $\vali \in \typespacei$,
\begin{align*}
    \big|\, \alg(\mathcal S, i, \vali, \strats) - \utili(\vali, \strati(\vali), \stratsmi) \,\big| < \eps. 
\end{align*}
\end{definition}

We aim to estimate the interim utility in the above definition, instead of the ex ante utility $\Ex[\vals \sim \dists]{\expostUi(\vali, \strats(\vals))}$, because the ex ante utility is difficult to estimate due to the randomness of bidder $i$'s own value. Even in a first-price auction with a single bidder, the bidder's ex ante utility $\E_{\vali \sim \disti}[\vali - \strati(\vali)]$ cannot be estimated using finitely many samples for all distributions $\disti$ and for all strategies $\strati$ simultaneously. The interim utility $\utili(\vali, \bidi, \stratsmi) = \Ex[\valsmi\sim \distsmi]{\expostUi(\vali, \bidi, \stratsmi(\valsmi))}$, on the other hand, does not involve randomization over bidder $i$'s own value and is easier to estimate. 

We note that the utility estimation problem cannot be solved if $\strategyset$ contains all possible bidding strategies, including non-monotone and monotone ones. 
\begin{proposition}\label{prop:utility-estimation-unsolvable}
The utility estimation problem cannot be solved with finitely many samples if $\strategyset$ contains all possible bidding strategies.
\end{proposition}
\begin{proof}
Consider an auction with two bidders, the first bidder having value~$1$ and bidding~$\frac 1 2$, and the second bidder's value~$\vali[2]$ uniformly drawn from $[0, 1]$.
Any finite set of samples of~$\vali[2]$ has probability measure~$0$ in the distribution of~$\vali[2]$.
Therefore on any set of samples, there are bidding strategies of bidder~$2$ that look the same on the sampled values but give bidder~$1$ drastically different utilities in expectation on the value distribution.
\end{proof}

Given the above impossibility result, we restrict $\strategyset$ to be the set of monotone joint bidding strategies.   
This is without loss of generality according to Proposition~\ref{prop:monotone}.

\subsection{Upper Bound of Sample Complexity of Utility Estimation}
\label{FPA_sample:fpa-upper-bound}
In this subsection, we show that $\tilde O(n / \eps^2)$ value samples suffice for estimating the interim utilities for all monotone bidding strategies.  
The estimation algorithm is the empirical distribution estimator, which outputs the expected utility on the uniform distribution over the samples. 


\begin{definition}
\label{def:emp}
The \emph{empirical distribution estimator}, denoted by $\emp$, 
estimates interim utilities on the uniform distribution over the samples.
Formally, on samples $\mathcal S = \{\vals^{(1)}, \ldots, \vals^{(\nsample)}\}$, for bidder~$i$ with value~$\vali$, for joint bidding strategy $\strats$,
\begin{equation*}
\emp(\mathcal S, i, \vali, \strats) := \frac{1}{\nsample} \sum_{j=1}^\nsample  \expostUi\big(\vali, \strati(\vali), \stratsmi(\valsmi^{(j)})\big).
\end{equation*}
\end{definition}

We now state an upper bound on the sample complexity of utility estimation, which is $\tilde O(\frac{H^2}{\eps^2} n)$ when ignoring logarithmic factors. 
\begin{theorem}
[Utility estimation by empirical distribution]
\label{thm:util-learn-upper-bound}
Suppose $\typespacei\subseteq[0, H]$.
For any $\eps>0, \delta \in (0, 1)$, there is
\begin{equation*}
M = O \left(\frac{H^2}{\eps^2} \left[n\log n\log \left(\frac{H}{\eps} \right) + \log\left(\frac{n}{\delta}\right)\right]\right),
\label{eq:util-learn-upper-bound}
\end{equation*}
such that for any $\nsample \geq M$, 
the empirical distribution estimator $\emp$ $(\eps, \delta)$-estimates with $\nsample$ samples
the utilities over the set of all monotone bidding strategies.
\end{theorem}





\subsubsection{Pseudo-Dimension and the Proof of Theorem \ref{thm:util-learn-upper-bound}}
\label{sec:pseudodim}

To prove Theorem \ref{thm:util-learn-upper-bound}, we use a tool called \emph{pseudo-dimension} from stasitical learning theory (see, e.g., \citep{anthony2009neural}), which captures the complexity of a class of functions.  

\begin{definition}
    \label{def:pseudo-dimension}
    Let $\funcClass$ be a class of real-valued functions on input space $\PinputSpace$. A set of inputs $\Pinputi[1], \ldots, \Pinputi[m]$ is said to be \emph{pseudo-shattered} if there exist \emph{witnesses} $\Pwitnessi[1], \ldots, \Pwitnessi[m] \in \mathbb R$ such that for any label vector $\Plabels\in\{1, -1\}^m$, there exists $\func_{\Plabels}\in \funcClass$ such that $\sgn(\func_{\Plabels}(\Pinputi) - \Pwitnessi)  = \Plabeli$ for each $i=1, \ldots, m$, where $\sgn(y)=1$ if $y>0$ and $-1$ if $y<0$. The \emph{pseudo-dimension} of $\funcClass$, $\Pdim(\funcClass)$, is the size of the largest set of inputs that can be pseudo-shattered by $\funcClass$. 
\end{definition}

\begin{definition}
	\label{def:uniform-convergence}
	For $\eps>0, \delta \in (0, 1)$, a class of functions $\funcClass: \PinputSpace \to \mathbb R$ is \emph{$(\eps, \delta)$-uniformly convergent with sample complexity $M$} if for any $\nsample \geq M$, for any distribution $\dist$ on~$\PinputSpace$, if $x^{(1)}, \ldots, x^{(\nsample)}$ are i.i.d.\@ samples from~$\dist$, with probability at least $1 - \delta$, for every $\func \in \funcClass$,
	$ \big| \Ex[\Pinput \sim \dist]{\func(\Pinput)} - \frac 1 {\nsample} \sum_{j = 1}^{\nsample} \func(x^{(j)}) \big| < \eps$. 
\end{definition}

\begin{theorem}[See, e.g., \citep{anthony2009neural}]
	\label{thm:pseudo-dimension}
	Let $\funcClass$ be a class of functions with range $[0, H]$ and pseudo-dimension $d = \Pdim(\funcClass)$, 
	for any $\eps>0$, $\delta\in(0, 1)$, 
	$\funcClass$ is $(\eps, \delta)$-uniformly convergent with sample complexity $O\left( \frac{H^2}{\eps^2} \big[d\log(\frac{H}{\eps}) + \log(\frac{1}{\delta})\big] \right)$.
\end{theorem}

We prove Theorem~\ref{thm:util-learn-upper-bound} by treating the utilities on monotone bidding strategies as a class of functions, whose uniform convergence implies that $\emp$ learns the interim utilities.


For each bidder $i$, let $\func^{\vali, \strats}$ be the function that maps the opponents' values to bidder~$i$'s ex post utility: 
\begin{equation*}
    \func^{\vali, \strats}(\valsmi) = \expostUi(\vali, \strati(\vali), \stratsmi(\valsmi)).
\end{equation*}
Let $\funcClass_i$ be the set of all such functions corresponding to the set of monotone strategies, 
\begin{equation*}
    \funcClass_i = \big\{\func^{\vali, \strats }(\cdot) \given \vali \in \typespacei,~~ \strats \text{ is monotone} \big\}.
\end{equation*}
By Equation \eqref{eq:interim-util}, the expectation of $\func^{\vali, \strats}(\cdot)$ over~$\distsmi$ is the interim utility of bidder $i$: 
\begin{equation*}
    \Ex[\valsmi\sim \distsmi]{\func^{\vali, \strats}(\valsmi)} = \utili(\vali, \strati(\vali), \stratsmi).
\end{equation*}
By Definition~\ref{def:emp}, on samples $\mathcal S = \{ \vals^{(1)}, \ldots, \vals^{(\nsample)}\}$, $\emp(\mathcal S, i, \vali, \strats) = \frac 1 {\nsample} \sum_{j = 1}^{\nsample} \func^{\vali, \strats}(\vals^{(j)}_{-i})$. 
Thus, 
\begin{align}
  & \Big| \emp(\sampleSet, i, \vali, \strats) -  \utili(\vali, \strati(\vali), \stratsmi) \Big| \nonumber\\
 & = \left| \Ex[\valsmi\sim\distsmi]{\func^{\vali, \strats}(\valsmi)} - \frac 1 {\nsample} \sum_{j = 1}^{\nsample} \func^{\vali, \strats}(\vals^{(j)}_{-i})\right|.  \label{eq:emp-func}
\end{align}

The right hand side of~\eqref{eq:emp-func} is the difference between the expectation of $\func^{\vali, \strats}$ on the distribution $\distsmi$ and that on the empirical distribution with samples drawn from $\distsmi$.
Now by Theorem~\ref{thm:pseudo-dimension},
to bound the number of samples needed by $\emp$ to $(\eps, \delta)$-estimate the utilities over monotone strategies, 
it suffices to bound the pseudo-dimension of~$\funcClass_i$.
With the following key lemma, the proof of Theorem \ref{thm:util-learn-upper-bound} is completed by observing that the range of each $\func^{\vali, \strats}$ is within $[-H, H]$ and by taking a union bound over $i \in [n]$.



\begin{lemma}
\label{lem:pseudo-dimension-utility-class}
$\Pdim(\funcClass_i) = O(n \log n)$.
\end{lemma}

The proof of Lemma~\ref{lem:pseudo-dimension-utility-class} follows a powerful framework introduced by \citep{MR16} and \citep{BSV18} for bounding the pseudo-dimension of a class $\funcClass$ of functions: given inputs that are to be pseudo-shattered, fixing any witnesses, one classifies the functions in~$\funcClass$ into subclasses, such that the functions in the same subclass output the same label on all the inputs; by counting and bounding the number of subclasses, one can bound the number of shattered inputs.
Our proof follows this strategy.  To bound the number of subclasses, we make use of the monotonicity of bidding functions, which is specific to our problem.

\begin{proof}[Proof of Lemmea \ref{lem:pseudo-dimension-utility-class}]
By definition, given any $\vali$ and $\strats$ (with $\bidi = \strati(\vali)$), the output of $\func^{\vali, \strats}$ on input $\valsmi$ is
\begin{align*}
    \func^{\vali, \strats}(\valsmi) & = \vali \alloci(\bids) - \payi(\bids) \\
    & = \vali \alloci(\bids) - \alloci(\bids) f_i(b_i) - g_i(b_i) \\
    & = \big(\vali - f_i(b_i)\big) \alloci(\bids) - g_i(b_i). 
\end{align*}
Because the allocation $\alloci(\bids)$ is that the highest bidder wins with random tie breaking, $\func^{\vali, \strats}(\valsmi)$ must take one of the following $n+1$ values: 
\begin{align*}
    \vali - f_i(b_i) - g_i(b_i), ~ \tfrac{\vali - f_i(b_i)}{2} - g_i(b_i), ~ \ldots, ~ \tfrac{\vali - f_i(b_i)}{n}- g_i(b_i), ~ 0 - g_i(b_i).
\end{align*}
This value is fully determined by the $n-1$ comparisons $\bidi \lesseqqgtr \strati[j](\vali[j]^k)$, one for each $j\ne i$. 

Let $\valsmi^{(1)}, \ldots, \valsmi^{(m)}$ be any $m$ inputs. 
We argue that the function class $\funcClass_i$ can be divided into $O(\nsample^{2n})$ sub-classes $\{\funcClass_i^{\mathbf k }\}_{\mathbf k \in [m+1]^{2(n-1)}}$ 
such that each sub-class $\funcClass_i^{\mathbf k }$ generates at most $O(\nsample^n)$ different label vectors on the $m$ inputs. 
Thus $\funcClass_i$ generates at most $O(\nsample^{3n})$ label vectors in total. 
To pseudo-shatter $\nsample$~inputs, we need $O(\nsample^{3n})\ge 2^\nsample$, which implies $\nsample = O(n\log n)$. 

We now define sub-classes $\{\funcClass_i^{\mathbf k}\}_{\mathbf k}$, each indexed by $\mathbf k \in [m + 1]^{2(n-1)}$.  
For each dimension $j \in [n]\setminus\{i\}$, we sort the $\nsample$ inputs by their $j^{\text{-th}}$ coordinates non-decreasingly, and use $\pi(j, \cdot)$ to denote the resulting permutation over $\{1, 2, \ldots, \nsample\}$: 
formally, $\vali[j]^{(\pi(j, 1))} \le \vali[j]^{(\pi(j, 2))}\le \cdots \le \vali[j]^{(\pi(j, \nsample))}$.
For each function $\func^{\vali, \strats}(\cdot)$, for each $j$, 
we define two special positions: 
\begin{align*}
    & k_{j, 1} = \max \Big\{0, ~ \big \{k: \strati[j](\vali[j]^{(\pi(j, k))}) < \bidi \big\} \Big\}, \\
    & k_{j, 2}= \min \Big\{m + 1, ~ \big \{k: \strati[j](\vali[j]^{(\pi(j, k))}) > \bidi \big\} \Big\}.
\end{align*}
These two positions are well defined because $\strati[j](\cdot)$ is monotone. By definition, if $k_{j, 1} < k_{j, 2} - 1$, then for any $k$ such that $k_{j, 1} < k < k_{j, 2}$, we must have $\strati[j](\vali[j]^{(\pi(j, k))}) = \bidi$. 
We let a function $\func^{\vali, \strats}(\cdot)$ belong to the sub-class $\funcClass_i^{\mathbf k }$ where the index $\mathbf k$ is $(k_{j, 1}, k_{j, 2})_{j\in[n]\backslash\{i\}}$.  
The number of sub-classes is the number of indices, which is bounded by $(m+1)^{2(n-1)}$.

We now show that the functions within a sub-class $\funcClass_i^{\mathbf k}$ give rise to at most $(m+1)^n$ label vectors on the $m$ inputs.
Let us focus on one such class with index~$\mathbf k$. 
On the $k^{\text{-th}}$ input~$\valsmi^{(k)}$, 
a function's membership in  
$\funcClass_i^{\mathbf k}$ suffices to specify whether bidder~$i$ is a winner on this input, and, if so, the number of other bidders winning at a tie.
Therefore, the class index~$\mathbf k$ determines a mapping $c: [m] \to \{0, 1, \ldots, n\}$, with $c(k) > 0$ meaning bidder~$i$ is a winner on input~$\valsmi^{(k)}$ at a tie with $c(k)-1$ other bidders, and $c(k) = 0$ meaning bidder~$i$ is a loser on input~$\valsmi^{(k)}$.  
Then, the output of a function $\func^{\vali, \strats}(\cdot) \in \funcClass_i^{\mathbf k}$ on input~$\valsmi^{(k)}$ is $\frac{\vali - f_i(b_i)}{c(k)} - g_i(b_i)$ if $c(k) > 0$ and $-g_i(b_i)$ otherwise.  
The same utility is output on two inputs $\valsmi^{(k)}$ and~$\valsmi^{(k')}$ whenever $c(k) = c(k')$.
Consider the set~$S \subseteq [m]$ of inputs that are mapped to one integer by~$c$, and fix any $|S|$ witnesses. By varying the function in the subclass $\funcClass_i^{\mathbf k}$, we can generate at most $|S| + 1 \leq \nsample + 1$ patterns of labels on the input set $S$, because we are comparing the same utility with $|S|$ witnesses.
The label vector for the entire input set $[m]$ is the concatenation of these patterns of labels. Since the image of $c$ has $n+1$  integers, and there are at most $(\nsample+1)^{n+1}$ label vectors.

To conclude, the total number of label vectors generated by $\funcClass_i=\bigcup_{\mathbf k } \funcClass_i^{\mathbf k }$ is at most 
\[ (\nsample+1)^{2(n-1)} (\nsample+1)^{n+1} \le (\nsample+1)^{3n}. \]
To pseudo-shatter $\nsample$~inputs, we need $(\nsample+1)^{3n}\ge 2^\nsample$, which implies $\nsample=O(n\log n)$.
\end{proof}

\subsection{Lower Bound of Sample Complexity of Utility Estimation}
\label{FPA_sample:lower-bound}

We give an information-theoretic lower bound on the number of samples needed for any algorithm to estimate utilities over monotone strategies in a first-price auction. 
The lower bound matches our upper bound up to polylogarithmic factors.

\begin{theorem}
\label{thm:lower-bound-learning-util}
For any $\eps < \frac 1 {4000}, \delta < \frac 1 {20}$, 
there is a family of product value distributions for which no algorithm can $(\eps, \delta)$-estimate utilities over the set of all monotone bidding strategies with $m \leq \frac{1}{4\times10^8}\cdot \frac{n}{\eps^2}$ samples. 
\end{theorem}

The proof of Theorem~\ref{thm:lower-bound-learning-util} is in Appendix~\ref{sec:proof-thm:lower-bound-learning-utility}.  
As a sketch, the product value distributions we construct encode length $n-1$ binary strings by having a slightly unfair Bernoulli distribution for each bidder, the bias shrinking as $n$ grows large.
We then show that, if with a few samples a learning algorithm can estimate utilities for all monotone bidding strategies, then there must exist two product value distributions from the family that differ at only one coordinate, and yet they can be told apart by the learning algorithm. 
This must violate the well-known information-theoretic lower bound for distinguishing two distributions \citep{MansourNotes}.

\section{Sample Complexity of Learning BCE}
This section studies how to learn BCEs using samples from the value distribution $\dists$. 

\subsection{Estimating Utility by Empirical Product Distributions}
\label{sec:empp}
Section \ref{FPA_sample:fpa-upper-bound} shows that the empirical distribution estimator approximates interim utilities with high probability.  However, this does not immediately imply that the auction on the empirical distribution is a close approximation to the auction on the original distribution $\dists$. 
This is because the empirical distribution over samples is \emph{correlated} --- the values $\vals^{(j)} = (\vali[1]^{(j)}, \ldots, \vali[n]^{(j)})$ are drawn as a vector, instead of independently. 
The equilibria (either BCE or BNE) with respect to this correlated empirical distribution do not correspond to the equilibria on the original product distribution $\dists$. 
Therefore, it is desirable that utilities can also be estimated on a \emph{product} distribution arising from the samples, where each bidder's value is independently drawn, uniformly from the $\nsample$ samples of her value.  
We show that this can indeed be done, without a substantial increase in the sample complexity.
The key technical step, Lemma~\ref{lem:relation-uniform-convergence},
is a reduction from learning on empirical distribution to learning on empirical product distribution.



\begin{definition}
\label{def:empp}
Given samples $\bm x^{(1)}, \ldots, \bm x^{(\nsample)}$ from a product distribution $\dists = \prod_{i=1}^n \disti$,  
let $\empDisti$ be the uniform distribution over $\{x_i^{(1)}, \ldots, x_i^{(\nsample)}\}$. The \emph{empirical product distribution} is the product distribution
    $\empDists = \prod_{i=1}^n \empDisti$.
\end{definition}

\begin{definition}
\label{def:uniform-convergence-product}
    For $\eps>0, \delta \in (0, 1)$, a class of functions $\funcClass: \prod_{i=1}^n \typespacei \to \mathbb R$ is \emph{$(\eps, \delta)$-uniformly convergent on product distribution with sample complexity $M$} if for any $\nsample \geq M$, for any product distribution $\dists$ on~$\prod_{i=1}^n \typespacei$, 
    if $\bm x^{(1)}, \ldots, \bm x^{(\nsample)}$ are i.i.d.\@ samples from~$\dists$, 
    with probability at least $1 - \delta$, for every $\func \in \funcClass$,
    \begin{equation*}
		\big| \Ex[\bm x \sim \dists]{\func(\bm x)} - \Ex[\bm x \sim \empDists]{\func(\bm x)} \big| ~ < ~ \eps,
    \end{equation*}
    where $\empDists = \prod_{i=1}^n \empDisti$ is the empirical product distribution. 
\end{definition}

\begin{lemma}
\label{lem:relation-uniform-convergence}
Let $\funcClass$ be a class of functions from a product space $\typespaces = \prod_{i=1}^n T_i$ to $[0, H]$. 
If $\funcClass$ is $(\eps, \delta)$-uniformly convergent with sample complexity $\nsample(\eps, \delta)$, then $\funcClass$ is $\left(2\eps, \frac{H\delta}{\eps}\right)$-uniformly convergent on product distribution with sample complexity $\nsample(\eps, \delta)$. 
In other words, $\funcClass$ is $(\eps', \delta')$-uniformly convergent with sample complexity $\nsample(\frac{\eps'}{2}, \frac{\delta'\eps'}{2H})$.
\end{lemma}
\noindent Lemma~\ref{lem:relation-uniform-convergence} is closely related to a concentration inequality by \citep{DHP16}, who show that for any \emph{single} function $h:\typespaces\to[0, H]$, the expectation of $h$ on the empirical product distribution $\empDists$ is close to its expectation on the original distribution $\dists$ with high probability.
Our lemma generalizes \citep{DHP16} to the simultaneous concentration for a family of functions,   
and seems more handy for applications such as ours.
We believe Lemma \ref{lem:relation-uniform-convergence} is of broader interest beyond the learnability of equilibrium in auctions; it might be useful for the study of sample complexity for other stochastic optimization problems with multiple independently distributed random variables, such as prophet inequalities and Pandora's box problem \citep{guo_generalizing_2021}. 
In fact, the preliminary version of our work \citep{fu_learning_2020} uses Lemma~\ref{lem:relation-uniform-convergence} to derive the sample complexity for the Pandora's Box problem.

\begin{proof}[Proof of Lemma \ref{lem:relation-uniform-convergence}]
Write the $m$ samples $\sampleSet = \{\bm x^{1}, \ldots, \bm x^{m}\}$ from $\dists$ as an $\nsample \times n$ matrix $(x_i^{j})$, where each row $j\in[\nsample]$ represents a sample $\bm x^{j}$, and each column~$i\in[n]$ consists of the $m$ values sampled from~$\disti$. 
Then, we draw $n$ permutations $\pi_1, ..., \pi_n$ of $[\nsample]=\{1, \ldots, \nsample\}$ independently and uniformly at random, and permute the $\nsample$ elements in column~$i$ by~$\pi_i$. 
Regard each new row $j$ as a new sample, denoted by $\permSamples^{j} = (\samplei[1]^{\pi_1(j)}, \samplei[2]^{\pi_2(j)}, ..., \samplei[n]^{\pi_n(j)})$. 
Given $\pi_1, \ldots, \pi_n$, the ``permuted samples'' $\{\permSamples^{1}, \ldots, \permSamples^{m}\}$ have the same distributions as $\nsample$ i.i.d.\@ random draws from~$\dists$. 

For $\func \in \funcClass$, let $p_\func = \Ex[\samples \sim\dists]{h(\samples )}$. 
By the definition of $(\eps, \delta)$-uniform convergence (not on product distribution),
\begin{equation}\label{eq:samples_pi}
\Prx[\sampleSet, \pi]{\exists \func \in \funcClass,\ \Big| p_\func - \frac{1}{\nsample }\sum_{j=1}^{\nsample} \func(\permSamples^j) \Big|\ge\eps} \le \delta.
\end{equation}

For a set of fixed samples $\sampleSet = (\samples^1, \ldots, \samples^\nsample)$, recall that $\empDisti[i]$ is the uniform distribution over $\{\samplei[i]^{1}, \ldots, \samplei[i]^{\nsample}\}$, and $\empDists = \prod_{i=1}^n \empDisti[i]$. 
We show that the expected 
value of $\func$ on $\empDists$ satisfies $\Ex[\samples\sim\empDists]{\func(\samples)} = \Ex[\pi]{\frac{1}{\nsample}\sum_{j=1}^\nsample \func(\permSamples^j)}$. This is because
\begin{align*}
    \Ex[\pi]{\frac{1}{\nsample}\sum_{i=1}^\nsample \func(\permSamples^j)}
    & = \frac{1}{\nsample} \sum_{j=1}^{\nsample} \Ex[\pi]{\func(\permSamples^j)} \\
    & = \frac{1}{\nsample}\sum_{j=1}^\nsample \sum_{(k_1, \ldots, k_n)\in[\nsample]^n} \func(\samplei[1]^{k_1}, \ldots, \samplei[n]^{k_n})  \cdot \Prx[\pi]{\pi_1(j)=k_1, \ldots, \pi_n(j)=k_n}  \\
    & =\frac{1}{\nsample}\sum_{j=1}^\nsample \sum_{(k_1, \ldots, k_n)\in[\nsample]^n} \func(\samplei[1]^{k_1}, \ldots, \samplei[n]^{k_n})\cdot \frac{1}{\nsample^n}\\
    & = \frac{1}{\nsample^n} \sum_{(k_1, \ldots, k_n)\in[\nsample]^n} \func(\samplei[1]^{k_1}, \ldots, \samplei[n]^{k_n}) \\
    & =\Ex[\samples \sim \empDists]{\func(\samples)}.
\end{align*}

Thus,   
\begin{align*}
     \big| p_\func - \Ex[\samples\sim\empDists]{\func(\samples)}\big|
     & = \left| p_\func - \Ex[\pi]{\frac{1}{\nsample}\sum_{j=1}^\nsample \func(\permSamples^j)} \right| \\
    & \le \Ex[\pi]{ \Big| p_\func - \frac{1}{\nsample}\sum_{j=1}^\nsample \func(\permSamples^j) \Big|}\\
    & \le \Prx[\pi]{ \Big| p_\func - \frac{1}{\nsample}\sum_{j=1}^\nsample \func(\permSamples^j) \Big|\ge \eps}\cdot H \\
    & \hspace{1.5em} + \left(1-\Prx[\pi]{\Big| p_\func - \frac{1}{\nsample}\sum_{j=1}^\nsample \func(\permSamples^j) \Big|\ge \eps}\right)\cdot\eps \\
    & \le \Prx[\pi]{\mathrm{Bad}(\func, \pi, \sampleSet)}\cdot H + \eps, 
\end{align*}
where we define event
\[ \mathrm{Bad}(\func, \pi, \sampleSet) = \mathbb{I}\left[\Big| p_\func - \frac{1}{\nsample}\sum_{j=1}^\nsample \func(\permSamples^j) \Big|\ge \eps\right].\]
We note that, whenever $\big| p_\func - \Ex[\vals\sim\empDists]{\func(\vals)}\big| \ge 2\eps$, we have $\Prx[\pi]{\mathrm{Bad}(\func, \pi, \sampleSet)}\ge \frac{\eps}{H}$.

Finally, consider the random draw of samples $\sampleSet\sim\dists$, 
\begin{align*}
     & \Prx[\sampleSet]{\vphantom{\Big|} \exists \func\in \funcClass, \ \big| p_\func - \Ex[\vals \sim \empDists]{\func(\vals)} \big|\ge 2\eps} 
    \\
    & \le \Prx[\sampleSet]{\exists \func \in \funcClass, \ \Prx[\pi]{\mathrm{Bad}(\func, \pi, \sampleSet)}\ge \frac{\eps}{H} } 
    \\
    & \le \Prx[\sampleSet]{\Prx[\pi]{\exists \func \in \funcClass, \ \mathrm{Bad}(\func, \pi, \sampleSet)\text{ holds} } \ge \frac{\eps}{H}} \\
    & \le  \frac{H}{\eps}\Ex[\sampleSet]{\vphantom{\Big|}\Prx[\pi]{\exists \func \in \funcClass, \ \mathrm{Bad}(\func, \pi, \sampleSet)\text{ holds} }  } && \text{by Markov inequality} \\
    &  = \frac{H}{\eps}\Prx[\sampleSet, \pi]{ \vphantom{\Big|} \exists \func \in \funcClass, \ \mathrm{Bad}(\func, \pi, \sampleSet)\text{ holds} } \\
     & \le \frac{H\delta}{\eps} && \text{by \eqref{eq:samples_pi}}. 
\end{align*}
\end{proof}

Combining Theorem~\ref{thm:util-learn-upper-bound} with Lemma~\ref{lem:relation-uniform-convergence}, we derive a result of utility estimation by empirical product distribution.

\begin{definition}
The \emph{empirical product distribution estimator} $\empp$ estimates interim utilities of a bidding strategy on the empirical product distribution~$\empDists = \prod_{i=1}^n \empDisti$.  Formally, 
for bidder~$i$ with value~$\vali$, for bidding strategy profile $\strats$, $\empp(\sampleSet, i, \vali, \strats) := 
\Ex[\valsmi\sim \empDistsmi] { \expostUi(\vali, \strati(\vali), \stratsmi(\valsmi)) }$. 
\end{definition}

\begin{theorem}
[Utility estimation by empirical product distribution]
\label{thm:util-learn-upper-bound-product}
Suppose $\typespacei\subseteq[0, H]$.
Let $\dists$ be a product distribution on $\prod_{i=1}^n \typespacei$. 
For any $\eps>0, \delta \in (0, 1)$, there is
\begin{equation}
M = O \left(\frac{H^2}{\eps^2} \left[n\log n\log \left(\frac{H}{\eps} \right) + \log\left(\frac{n}{\delta}\right)\right]\right),
\label{eq:util-learn-upper-bound-product}
\end{equation}
such that for any $\nsample \geq M$, 
the empirical product distribution estimator $\empp$ $(\eps, \delta)$-estimates with $\nsample$ samples
the utilities over the set of all monotone bidding strategies.
\end{theorem}


\subsection{Learning Equilibrium from Samples}
We are now ready to present our results for learning equilibria (BCE and BNE) using samples from the value distribution $\dists$. 
By Theorem~\ref{thm:util-learn-upper-bound-product}, utilities of the bidders on $\dists$ can be approximated by the utilities on the empirical product distribution $\empDists$, therefore the auctions on the two distributions share the same set of approximate equilibria:

\begin{theorem}[Learning BCE from samples]
\label{theorem:find-BCE}
Suppose $\typespacei\subseteq[0, H]$ and $\dists$ is a product distribution on $\prod_{i=1}^n \typespacei$. 
For any $\eps, \eps'>0, \delta \in (0, 1)$,
by drawing $m \ge \eqref{eq:util-learn-upper-bound-product}$ samples from $\dists$, with probability at least $1-\delta$, we have: 
Any monotone $\eps'$-BCE $Q$ on the empirical product distribution $\empDists = \prod_{i=1}^n \empDisti$ is a monotone $(\eps' + 2\eps)$-BCE on $\dists$.
Conversely, any monotone $\eps'$-BCE $Q$ on $\dists$  is a monotone $(\eps' + 2\eps)$-BCE on $\empDists$. 
Formally:
\[\mathrm{BCE}(\empDists, \eps') \subseteq \mathrm{BCE}(\dists, \eps'+2\eps)\, \text{ and }\,\mathrm{BCE}(\dists, \eps') \subseteq \mathrm{BCE}(\empDists, \eps'+2\eps).\] 
\end{theorem}

\begin{proof}
We will prove $\mathrm{BCE}(\empDists, \eps') \subseteq \mathrm{BCE}(\dists, \eps'+2\eps)$. The other direction is analogous. 
According to Theorem \ref{thm:util-learn-upper-bound-product}, for any bidder $i$ with value $\vali$ and bid $\bidi$, for any monotone strategies $\stratsmi$ of other bidders, bidder $i$'s interim utilities on distributions $\dists$ and $\empDists$ satisfy: 
\begin{equation}
    \label{eq:utility-preserve-product}
    \bigg| \E_{\valsmi \sim \distsmi}\Big[ \expostUi(\vali, \bidi, \stratsmi(\valsmi))\Big] - \E_{\valsmi \sim \empDistsmi} \Big[ \expostUi(\vali, \bidi, \stratsmi(\valsmi))\Big] \bigg| ~ \le ~ \eps. 
\end{equation}
Let $Q$ be any monotone $\eps'$-BNE $Q$ on $\empDists$. By Definition \ref{def:bce}, bidder $i$'s utility gain by deviating according to deviation function $\phi_i$ satisfies 
\begin{align*}
    \E_{\strats \sim Q} \bigg[ \E_{\valsmi \sim \empDistsmi} \Big[ \expostUi(\vali, \phi_i(\strati, \vali), \stratsmi(\valsmi)) - \expostUi(\vali, \strati(\vali), \stratsmi(\valsmi)) \Big] \bigg] ~ \le ~ \eps'. 
\end{align*}
Applying \eqref{eq:utility-preserve-product} to above, we obtain 
\begin{align*}
    \E_{\strats \sim Q} \bigg[ \E_{\valsmi \sim \distsmi} \Big[ \expostUi(\vali, \phi_i(\strati, \vali), \stratsmi(\valsmi)) - \expostUi(\vali, \strati(\vali), \stratsmi(\valsmi)) \Big] \bigg] ~ \le ~ \eps' + 2\eps, 
\end{align*}
which implies that $Q$ is an $(\eps' + 2\eps)$-BCE on $\dists$. 
\end{proof}

The main implication of Theorem~\ref{theorem:find-BCE} is the following: if a mediator wants to coordinate the bidders in a non-truthful auction (such as first-price auction) in an incentive-compatible way, but does not know the bidders' value distribution, the mediator can still achieve that by computing an approximate correlated equilibrium for the bidders using samples from the distribution.  Theorem~\ref{theorem:find-BCE} characterizes the number of samples needed. It is almost linear in the number of bidders $n$ and polynomial in the target approximation accuracy $\eps$, which is a statistically moderate sample complexity.

A similar conclusion also holds for BNEs: 
\begin{theorem}[Learning BNE from samples] \label{theorem:find-BNE}
Under the same condition as Theorem~\ref{theorem:find-BCE}, 
\[\mathrm{BNE}(\empDists, \eps') \subseteq \mathrm{BNE}(\dists, \eps'+2\eps)\,\text{ and }\,\mathrm{BNE}(\dists, \eps') \subseteq \mathrm{BNE}(\empDists, \eps'+2\eps).\]
\end{theorem}



Given some recent progress on the computation of BNE in first-price auctions on \emph{discrete} distributions \citep{wang_bayesian_2020, chen_complexity_2023}, 
we present an interesting corollary of Theorem~\ref{theorem:find-BNE}: if there exists an algorithm that can compute BNE for a first-price auction on discrete value distributions, then there also exists an algorithm that can compute approximate BNE on \emph{any} distributions, by simply sampling from the distribution and running the former algorithm on the empirical product distribution (which is discrete). 
\begin{corollary}
\label{cor:polytime-equilibria}
If there exists an algorithm that computes monotone $\eps'$-BNE for the first-price auction on any discrete product value distributions $\bm{D} = \prod_{i=1}^n D_i$, then there exists a sample-access algorithm that computes $(\eps'+\eps)$-BNE for the first-price auction on any product distributions $\bm F = \prod_{i=1}^n F_i$ with high probability.
If the running time of the former algorithm is polynomial in $\frac{1}{\eps'}$ and the support size of each discrete $D_i$, then the running time of the latter algorithm is $\mathrm{poly}(\frac{1}{\eps}, \frac{1}{\eps'})$, which does not depend on the support size of $F_i$ (and $F_i$ can be continuous). 
\end{corollary}

\section{Discussion}
In this work, we obtained the first sample complexity result for learning strategic-form Bayesian correlated equilibria in non-truthful auctions such as first-price and all-pay auctions. En route, we showed that bidders' expected utilities can be estimated using a moderate amount of value samples for all monotone bidding strategies. 
Such a moderate sample complexity suggests that learning to coordinate bidders in non-truthful auctions is statistically feasible. 
Our work can be a starting point for several future research directions:
\begin{itemize}
\item{\it Other types of BCE.}
The learnability of strategic-form BCE in non-truthful auctions relies on its simple form: recommending monotone joint bidding strategies to bidders without eliciting bidders' values.  Other types of BCEs, such as a communication equilibrium which includes an elicitation phase and a recommendation phase \citep{myerson_optimal_1982, forges_correlated_2006}, need not have a monotone structure, and we do not know whether they are efficiently learnable. 

\item{\it Correlated value distribution.}
Our results also depend on bidders' values being drawn independently.  
With correlated values, the conditional distribution of opponents' values changes with a bidder's own value, and any na\"ive utility estimation algorithm needs a number of samples that grows linearly with the size of a bidder's value space.
It is interesting whether there are meaningful tractable middle grounds between independent distributions and arbitrary correlated distributions.

\item{\it Multi-item auctions and general games.}
Monotonicity is a natural assumption on bidding strategies in a single-item auction, but it does not generalize to multi-parameter settings, where equilibria are difficult to characterize.
It is an interesting question whether our results can be generalized to multi-item auctions, such as simultaneous first-price auctions, via more general, lossless structural assumptions on the bidding strategies. 
One can ask an even more general question: when can BCE be learned from type samples in general incomplete-information games? 
\end{itemize}


\bibliographystyle{alpha} 
\bibliography{reference} 


\appendix

\section{Proof of Theorem \ref{thm:lower-bound-learning-util}}
\label{sec:proof-thm:lower-bound-learning-utility}
Fixing $\eps > 0$, fixing $\distconst = 2000$, we first define two value distributions.
Let $\posdist$ be a distribution supported on $\{0, 1\}$, and for $\val \sim \posdist$, $\Prx{\val = 0} = 1 - \frac{1 + \distconst \eps}{n}$, and $\Prx{\val = 1} = \frac{1 + \distconst \eps}{n}$.  
Similarly define $\negdist$: for $\val \sim \negdist$, $\Prx{\val = 0} = 1 - \frac{1 - \distconst \eps}{n}$, and $\Prx{\val = 1} = \frac{1 - \distconst \eps}{n}$.   

Let $\kl(\posdist; \negdist)$ denote the KL-divergence between the two distributions.
\begin{claim}
\label{cl:lb-kl}
$\kl(\posdist; \negdist)= O(\frac {\eps^2}{n})$.
\end{claim}
\begin{proof}
By definition,
\begin{align*}
     \kl(\posdist; \negdist) 
    ={}& \frac{1 + \distconst\eps}{n} \ln \left( \frac{1 + \distconst\eps}{1 - \distconst\eps} \right) 
      + \frac{n - 1 - \distconst \eps}{n} \ln \left(\frac{n - 1 - \distconst \eps}{n - 1 +\distconst \eps}\right) \\
   ={}& \frac 1n  \ln \left( \frac {1 + \distconst\eps}{1 - \distconst \eps} \cdot \frac{(1 - \frac{\distconst \eps}{n - 1})^{n-1}}{(1 + \frac{\distconst \eps}{n-1})^{n-1}}
   \right)
   + \frac {\distconst \eps}{n} \ln \left(\frac{1 + \distconst \eps}{1 - \distconst \eps} \cdot 
   \frac{1 + \frac{\distconst \eps}{n-1}}{1 - \frac{\distconst \eps}{n-1}}\right) \\
   \leq{}& \frac 1n  \ln \left(  \frac {1 + \distconst\eps}{1 - \distconst \eps} \cdot \frac{\left(1 - \frac{\distconst \eps}{n - 1}\right)^{n-1}}{1 + \distconst \eps} \right)
    + \frac {2\distconst \eps}{n} \ln \left(1 + \frac{2\distconst \eps}{1 - \distconst \eps} \right)
\\
   \leq{}& \frac 1 n \ln \left( 
   \frac{1 - \distconst \eps + \frac 1 2 (\distconst \eps)^2}{1 - \distconst \eps}
   \right) + \frac{8\distconst^2 \eps^2}{n} \\
   \leq{}& \frac{10\distconst^2 \eps^2}{n}.
\end{align*}
In the last two inequalities we used $\distconst \eps < \frac 1 2$ and $\ln (1+x) \leq 1+x$ for all $x > 0$.
\end{proof}

It is well known that an upper bound on KL-divergence implies an information-theoretic lower bound on the number of samples to distinguish two distributions (e.g., \citep{MansourNotes}).
\begin{corollary}
\label{cor:lb-kl-single}
Given $t$ i.i.d.\@ samples from $\posdist$ or~$\negdist$, if $t \leq \frac{n}{80\distconst^2 \eps^2}$, no algorithm~$\diffalg$ that maps samples to $\{\posdist, \negdist\}$ can do the following: when the samples are from~$\posdist$, $\diffalg$ outputs~$\posdist$ with probability at least $\frac 2 3$, and if the samples are from~$\negdist$, $\diffalg$ outputs~$\negdist$ with probability at least~$\frac 2 3$.  
\end{corollary}

We now construct product distributions using $\posdist$ and~$\negdist$.  
For any $S \subseteq [n - 1]$, define product distribution $\dists_S$ to be $\prod_i \disti$ where $\disti = \posdist$ if $i \in S$, and $\disti = \negdist$ if $i \in [n-1] \setminus S$, and $F_n$ is a point mass on value~$1$.
For any $j \in [n - 1]$ and $S \subseteq [n - 1]$, distinguishing $\dists_{S \cup \{j\}}$ and $\dists_{S \setminus \{j\}}$ by samples from the product distribution is no easier than distinguishing $\posdist$ and $\negdist$, because the coordinates of the samples not from $\disti[j]$ contains no information about~$\disti[j]$.  

\begin{corollary}
\label{cor:lb-kl}
For any $j \in [n - 1]$ and $S \subseteq [n - 1]$, given $t$ i.i.d.\@ samples from $\dists_{S \cup \{j\}}$ or $\dists_{S \setminus \{j\}}$, if $t \leq \frac n {80 \distconst^2 \eps^2}$, no algorithm~$\diffalg$ can do the following: when the samples are from $\dists_{S \cup \{j\}}$, $\diffalg$ outputs~$\dists_{S \cup \{j\}}$ with probability at least $\frac 2 3$, and when the samples are from $\dists_{S \setminus \{j\}}$, $\diffalg$ outputs~$\dists_{S \setminus \{j\}}$ with probability at least~$\frac 2 3$.
\end{corollary}

We now use Corollary~\ref{cor:lb-kl} to derive an information-theoretic lower bound on estimating utilities for monotone bidding strategies, for distributions in $\{\dists_S\}_{S \subseteq [n]}$.

\begin{proof}[Proof of Theorem~\ref{thm:lower-bound-learning-util}]
Without loss of generality, assume $n$ is odd.  
Let $S$ be an arbitrary subset of~$[n - 1]$ of size either $\lfloor n/2 \rfloor$ or $\lceil n/ 2 \rceil$.
We focus on the interim utility of bidder~$n$ with value~$1$ and bidding $\frac 1 2$.  
Denote this bidding strategy by~$\strati[n]$.
The other bidders may adopt one of two bidding strategies.
One of them is $\posbid$: $\posbid(0) = 0$ and $\posbid(1) = \frac 1 2 + \eta$ for sufficiently small $\eta > 0$.  
The other bidding strategy $\negbid(\cdot)$ maps all values to~$0$.
For $T \subseteq [n-1]$, let $\strats_T$ be the profile of bidding strategies where $\strati = \posbid$ for bidder $i \in T$, and $\strati = \negbid$ for bidder $i \in [n-1] \setminus T$.

For the distribution $\dists_S$, the interim utility of bidder $n$ is 
\begin{align*}
     & \utili[n]\left(1, \frac 1 2, \strats_T\right) 
    \frac 1 2  \Prx{\max_{i \in T} \vali = 0} \\
    & =  \frac 1 2 
    \left(1 - 
    \frac{1 + \distconst \eps}{n}
    \right)^{|S \cap T|}
    \left(1 - 
    \frac{1 - \distconst \eps}{n}
    \right)^{|T\setminus S|}
    \\
    & = \frac 1 2
    \left(
    1 - \frac{1 + \distconst \eps}{n}
    \right)^{|T|}
    \left(
    \frac{n - 1 + \distconst \eps}{n - 1 - \distconst \eps}
    \right)^{|T \setminus S|}.
\end{align*}
Therefore, for $T, T' \subseteq [n-1]$ with $|T| = |T'| $,
\begin{align*}
     \frac{\utili[n](1, \frac 1 2, \strats_T)}{\utili[n](1, \frac 1 2, \strats_{T'})} ={}& \left(
    1 + \frac{2\distconst \eps / (n-1)}{1 - \frac{\distconst \eps}{n-1}} 
    \right)^{|T\setminus S| - |T' \setminus S|} \\
     \geq{}& 1 + \frac{2\distconst \eps}{n-1} \cdot (|T \setminus S| - |T' \setminus S|);
\end{align*}
Suppose $|T \setminus S| \geq |T' \setminus S|$ and $|T| = |T'| \geq \lfloor \frac n 2 \rfloor$, then
\begin{align}
 & \utili[n]\left(1, \frac 1 2, \strats_T\right) - \utili[n]\left(1, \frac 1 2, \strats_{T'}\right) \\
 & \ge (|T \setminus S| - |T' \setminus S|) \cdot \frac {2\distconst \eps}{n-1} \cdot \utili[n]\left(1, \frac 1 2, \strats_{T'} \right) 
\notag \\
 & \geq (|T \setminus S| - |T' \setminus S|) \cdot \frac {2\distconst \eps}{n-1} \cdot \frac 1 {8 e^2}, 
 \label{eq:util-diff-eps}
\end{align}
where the last inequality is because $\utili[n](1, \frac 1 2, \strats_{T'}) \ge \frac{1}{2} (1 - \frac{2}{n})^n = \frac{1}{2} [(1 - \frac{2}{n})^\frac{n}{2}]^2\ge \frac{1}{2}  (\frac{1}{2e})^2 = \frac{1}{8e^2}$. 

Now suppose an algorithm~$\alg$ $(\eps, \delta)$-estimates the utilities of all monotone bidding strategies with $t \leq \frac{n}{80\distconst^2 \eps^2}$ samples~$\mathcal S$.
Define $\diffalg: \mathbb R_+^{n \times t} \times \mathbb N \to 2^{[n-1]}$ be a function that outputs, among all $T\subseteq [n-1]$ of size~$k$, the one that maximizes bidder~$n$'s utility when other bidders bid according to strategy $\strats_T$.  
Formally, 
\begin{align*}
    \diffalg(\sampleSet, k) = \argmax_{T \subseteq [n-1], |T| = k} \alg\left(\sampleSet, n, 1, (\strats_{T}, \strati[n]) \right),
\end{align*}

By Definition~\ref{def:util-learn-ensemble}, for any $S$ with $|S| = \lfloor n / 2 \rfloor$, for samples drawn from~$\dists_S$, with probability at least $1 - \delta$,
\begin{equation*}
    \alg(\sampleSet, n, 1, (\strats_{[n-1]\setminus S}, \strati[n]) ) \\
    \geq \utili[n]\left(1, \frac 1 2, \strats_{[n-1] \setminus S} \right) - \eps;
\end{equation*}
and for any $T \subseteq[n-1]$ with $|T| = \lceil n / 2 \rceil$,
\begin{equation*}
    \alg(\sampleSet, n, 1, (\strats_T, \strati[n]) )
    \leq \utili[n]\left(1, \frac 1 2, \strats_T \right) + \eps.
\end{equation*}
Therefore, for $W = \diffalg(\sampleSet, \lceil n / 2 \rceil)$, 
\begin{align*}
    \utili[n]\left(1, \frac 1 2, \strats_W \right) \geq 
    \utili[n] \left(1, \frac 1 2, \strats_{[n-1]\setminus S} \right) - 2\eps.
\end{align*}
Since $|W| = [n-1]\setminus S = \lceil n / 2 \rceil$, by \eqref{eq:util-diff-eps},
\begin{align*}
    \left(\lceil \frac n 2 \rceil - |W \setminus S| \right) \cdot \frac{\distconst \eps}{(n-1)4e^2} \leq 2\eps.
\end{align*}
So
\begin{align*}
    |W \cap S| \leq (n - 1) \cdot \frac{8e^2}{\distconst}.
\end{align*}
In other words, with probability at least $ 1- \delta$, $\diffalg(\sampleSet, \lceil n / 2 \rceil)$ is the complement of~$S$ except for at most $\frac{8e^2}{\distconst}$ fraction of the coordinates in $[n-1]$.

Similarly, for $S$ of cardinality $\lceil n / 2 \rceil$, 
\begin{align*}
    |\diffalg(\sampleSet, \lceil n / 2 \rceil) \cap S| \leq (n - 1) \cdot \frac{8e^2}{\distconst} + 1.
\end{align*}
Take $\diffconst$ to be $\frac{8e^2}{\distconst}$. We have $\diffconst<\frac 1 {20}$. For all large enough $n$ and all $S$ of size~$\lfloor n / 2 \rfloor$ or $\lceil n / 2 \rceil$, with probability at least $1 - \delta$, $\diffalg(\sampleSet, \lceil n / 2 \rceil)$ correctly outputs the elements not in~$S$ with an exception of at most $\diffconst$ fraction of coordinates.

Let $\coordsets$ be the set of all subsets of $[n-1]$ of size either $\lceil n / 2 \rceil$ or $\lfloor n / 2 \rfloor$.
Consider any~$S \in \coordsets$.  
Let $\diffset(S) \subseteq [n-1]$ denote the set of coordinates whose memberships in~$S$ are correctly predicted by $\diffalg(\sampleSet, \lceil n / 2 \rceil)$ with probability at least $2/3$; that is, $i \in \diffset(S)$ if{f} with probability at least $2/3$, $\diffalg(\sampleSet, \lceil n / 2 \rceil)$ is correct about whether $i \in S$.  
Let the cardinality of~$|\diffset(S)|$ be $z(n-1)$. Suppose we draw coordinate $i$ uniformly at random from $[n-1]$, and independently draw $t$ samples $\sampleSet$ from $\dists_S$, then the probability that $\diffalg(\sampleSet, \lceil n / 2 \rceil)$ is correct about whether $i\in S$ satisfies:
\begin{align*}
     & \Pr_{i, \sampleSet}\Big[ \diffalg(\sampleSet, \lceil n / 2 \rceil)\text{ is correct about whether }i\in S \Big] \\
     & \geq (1 - \diffconst) (1 - \delta)  \geq 0.9,
\end{align*}
and 
\begin{align*}
    & \Pr_{i, \sampleSet}\Big[ \diffalg(\sampleSet, \lceil n / 2 \rceil) \text{ is correct about whether }i\in S \Big] \\
    & \le \Prx[i]{i\in \diffset(S)}\cdot 1 + \Prx[i]{i\notin \diffset(S)}\cdot\frac{2}{3}  \\
    & = z\cdot 1 + (1-z)\cdot \frac 2 3,
\end{align*} 
which implies $z > 0.6$.  
If a pair of sets $S$ and~$S'$ differ in only one coordinate~$i$, and $i \in \diffset(S) \cap \diffset(S')$, then $\diffalg(\cdot)$ serves as an algorithm that tells apart $\dists_S$ and~$\dists_{S'}$, contradicting Corollary~\ref{cor:lb-kl}.  
We now show, with a counting argument, that such a pair of $S$ and~$S'$ must exist.

Since for each $S \in \coordsets$, $|\diffset(S)| \geq 0.6(n-1)$, there exists a coordinate $i \in [n-1]$ and $\mathcal T \subseteq \coordsets$, with $|\mathcal T| \geq 0.6 |\coordsets|$, such that for each $S \in \mathcal T$, $i \in \diffset(S)$. 
But $\coordsets$ can be decomposed into $|\coordsets| / 2$ pairs of sets, such that within each pair, the two sets differ by one in size, and precisely one of them contains coordinate~$i$.  
Therefore among these pairs there must exist one $(S, S')$ with $S, S' \in \mathcal T$, i.e., $i \in \diffset(S)$ and $i \in \diffset(S')$.
Using $\diffalg$, which is induced by~$\alg$, we can tell apart $\dists_S$ and~$\dists_{S'}$ with probability at least $2/3$, which is a contradiction to Corollary~\ref{cor:lb-kl}.
This completes the proof of Theorem~\ref{thm:lower-bound-learning-util}.
\end{proof}

\end{document}